\documentclass[a4paper,final]{amsart}

\usepackage{amsmath,amssymb,enumerate}
\usepackage{times}
\usepackage{graphicx}
\usepackage[caption=false]{subfig}
\usepackage{float}
\usepackage{a4wide}
\usepackage{xcolor,xspace}
\usepackage{hyperref}
\usepackage{cleveref}

\newtheorem{theorem}{Theorem}
\newtheorem{lemma}[theorem]{Lemma}

\newtheorem{proposition}[theorem]{Proposition}

\newtheorem*{theoremMain}{Main Theorem}

\newcommand{\doi}[1]{\href{http://dx.doi.org/#1}{\texttt{doi:#1}}}

\DeclareMathOperator{\p}{{\chi}}

\DeclareMathOperator{\dist}{d}

\newcommand\enc{\mathop{\mathrm{Enc}}}
\newcommand\enco{\mathop{\mathrm{Enc_0}}}

\newcommand \allp{\Sigma}
\newcommand \oddp{\Omega}
\newcommand \conp{\Gamma}
\newcommand \balp{{\mathcal B}}
\newcommand \dimp{\Upsilon}

\numberwithin{equation}{section}

\begin{document}

\title{Generalized Gray codes with prescribed ends of small dimensions
}
\thanks{Support from the Czech Science Foundation grant GA14-10799S is gracefully acknowledged}
\dedicatory{V\'aclav Koubek passed away before this work was completed. \\
The remaining  author deeply regrets this loss.}

\author{  Tom\'a\v{s} Dvo\v{r}\'ak}
\address{Faculty of Mathematics and Physics, Charles University, Prague, Czech Republic}
\author{V\'aclav Koubek}

\begin{abstract}
Given pairwise distinct vertices $\{\alpha_i , \beta_i\}^k_{i=1}$ of the $n$-dimensional hypercube $Q_n$ such that the distance $\dist(\alpha_i,\beta_i)$ is odd, are there paths $P_i$ between $\alpha_i$ and $\beta_i$ such that $\{V (P_i)\}^k_{i=1}$ partitions $V(Q_n)$? A~positive solution for every $n\ge1$ and $k=1$ is known as a Gray code of dimension $n$. 
In this paper we settle this problem for small values of $n$. As a corollary, we obtain that such spanning paths with prescribed endvertices exist for every $k<n$ unless $n = k + 1 = 4$ and the subgraph of $Q_4$ induced by $\{u_i , v_i\}^3_{i=1}$ forms a cycle of length six. This result is optimal in the sense that for every $n>2$ there is an instance of this problem with $n$ prescribed pairs for which a positive solution does not exists.  
\end{abstract}



\maketitle
\section{Introduction}
The $n$-dimensional hypercube $Q_n$ is an undirected graph whose vertex set consists of all $n$-bit strings, two vertices being joined by an edge whenever the corresponding strings differ in a single bit. A~systematic study of this class of graphs was initiated in early seventies of the last century by a series of papers written by Czech researchers \cite{H,HaMo72,Ne74}.
An important property of hypercubes is the fact that $Q_n$ is Hamiltonian for every $n\ge2$. It should be noted that in the literature, Hamiltonian paths and cycles in hypercubes are often disguised as Gray codes or cyclic Gray codes, respectively, that are popular due to their numerous applications in computer science \cite{Knuth,Sa}. 

It is well-known that  given two vertices $\alpha,\beta$ of $Q_n$, a~Hamiltonian path between $\alpha$ and $\beta$ exists iff the distance $d(\alpha,\beta)$ of $\alpha$ and $\beta$ is odd \cite{H}. Caha and Koubek in \cite{CK} suggested the  following generalization: a~set $\{\alpha_i ,\beta_i\}_{i=1}^k$ of pairwise disjoint pairs of vertices of $Q_n$ is called \emph{connectable} if there are paths $P_{\alpha_i\beta_i}$ between $\alpha_i$ to $\beta_i$ for all $i=1,\dots,k$ such that
each vertex of $Q_n$ occurs in exactly one of these paths.
In \cite{CK}  they showed that $\{\alpha_i,\beta_i\}_{i=1}^k$ with $d(\alpha_i,\beta_i)$ odd is connectable for every $n\ge2$ provided that $k\le(n-1)/3$, and asked about the maximum value of $k$ for which this statement holds.  

The connectability of $\{\alpha_i,\beta_i\}_{i=1}^k$ for $k=2$ follows from earlier works \cite{CK,D,LW}, while instances of the problem for $k$ bounded by a small constant were systematically studied by Casta\~neda and Gotchev  \cite{CG,CGL}. The authors of this paper and Gregor showed that when $k$ is unbounded, the problem becomes NP-hard or NP-complete, depending on the instance description \cite{DG, DK}. Moreover, in \cite{DGK17} they characterized the connectability of  $\{\alpha_i,\beta_i\}_{i=1}^k$ for $k<n$, but the proof is computer-assisted: the connectability of instances for $n\le5$ was verified by a computer search.    

The purpose of this paper is to settle the connectability problem for small values of $n$ without computer assistance. As a corollary, we obtain a computer-free proof of the following result \cite{DGK17}. 
\begin{theoremMain}
Let $A=\{\alpha_i,\beta_i\}_{i=1}^k\subseteq V(Q_n)$,  $d(\alpha_i,\beta_i)$ be odd for every $i=1,2,\dots,k$, and $k<n>0$.
Then $A$ is connectable unless $n=4$ and the subgraph of $Q_4$ induced by $A$ forms a 6-cycle.
\end{theoremMain}
As observed in \cite{DGK05}, the above result is optimal in the following sense: for every $n>2$ there are $n$ pairwise disjoint pairs $\alpha_i,\beta_i\in V(Q_n)$ with $d(\alpha_i,\beta_i)$ odd for all $i$ which are not  connectable.  Indeed, let $\{\alpha_i\}_{i=1}^n$ be all $n$ neighbors of a vertex $\alpha$ of $Q_n$ and $\{\beta_1,\dots,\beta_n\}\subseteq V(Q_n)\setminus\{\alpha_1,\dots,\alpha_n,\alpha\}$ be arbitrary such that $\dist(\alpha_i,\beta_i)$ is odd for all $i\in\{1,2,\dots,n\}$. If $\{\alpha_i,\beta_i\}_{i=1}^n$ were connectable, then $\alpha$ would be included in a path $P_{\alpha_i\beta_i}$between $\alpha_i$ and $\beta_i$ for some $i$. But this sequence has to pass through some $\alpha_j$, $j\ne i$, which means that $\alpha_j$ is included in both $P_{\alpha_i\beta_i}$ and $P_{\alpha_j\beta_j}$, contrary to  our assumption.


\section{Concepts and notation}
The terminology and notation used in this paper mostly follows \cite{DGK17}. In the rest of this text, $n$ always denotes a positive integer while 
$[n]$ stands for the set $\{0,1,\dots,n-1\}$.

\subsection{Hypercubes}
Let $V_n$ be the set of all $n$-bit strings.   
For $\alpha\in V_n$ and $i\in[n]$, let $\alpha (i)$ denote the $i$-th bit of $\alpha$.
For strings $\alpha ,\beta
\in V_n$, $\alpha\oplus\beta$ denotes the bitwise addition modulo 2 while $\Delta(\alpha,\beta)=\{i\in[n]\mid \alpha(i)\ne\beta(i)\}$.
Let $e^n_i\in V_n$ for $i\in[n]$ be the strings such that
$$e^n_i(j)=\begin{cases} 1&\text{ if }i=j,\\
0&\text{ if }i\ne j.\end{cases} $$
If $n$ is clear from the context, then we omit $n$ and write only $e_i$.

A \emph{hypercube} $Q_n$ of dimension $n$ is a graph with vertex set $V_n$, $ \alpha\beta $ being an edge  whenever $\alpha\oplus\beta =e_i$ for some $i\in[n]$.
The set of all edges of $Q_n$ is denoted by $E(Q_n)$. We use $\dist(\alpha,\beta)$ to denote the distance of $\alpha,\beta\in V(Q_n)$. Note that $\dist(\alpha,\beta)=|\Delta(\alpha,\beta)|$. Vertices $\alpha,\beta$ of $Q_n$ are called \emph{diametrical} if $\dist(\alpha,\beta)=n$.

In the following, we often need to transform $n$-bit strings into  those of lengths $n-1$ or $n+1$. For this purpose, we employ the following operations. For $i\in[n]$ and $\alpha\in V_n$ such that $\alpha (i)=k$ we define $\rho_{i=k}(\alpha )\in V_{n-1}$  as follows:
$$\rho_{i={k}}(\alpha )(j)=\begin{cases} \alpha (j)&\text{ if }j<i,\\
\alpha (j+1)&\text{ if }j\ge i\end{cases} $$
while $\rho_{i=k}(\alpha)$ is undefined if $\alpha(i)\ne k$.
Given a set $X\subseteq V_n$, $i\in[n]$ and $k\in[2]$, set $\rho_{i=k}(X)=\{\rho_{i=k}(\alpha )\mid\alpha\in X\}.$
Conversely, for a string  $\alpha\in V_{n-1}$, $i\in[n]$ and $k\in[2]$, let $\iota_{i=k}(\alpha )$ be a string of $V_n$ such that
$$\iota_{i=k}(\alpha )(j)=\begin{cases} \alpha (j)&\text{ if }j<i,\\
k&\text{ if }j=i,\\
\alpha (j-1)&\text{ if }j>i.\end{cases} $$
For $\alpha\in V_n$ we define the \emph{parity} $\chi(\alpha)$ of $\alpha$ by $\chi (\alpha)=\prod_{i\in[n]}-1^{\alpha(i)}$.
\par
\subsection{Pair-sets}
Let $P_2(V_n)=\{\alpha \beta\mid\alpha ,\beta\in V_n\}$ be the edge set of the complete graph with loops over the vertex set $V_n$. Note that elements of $P_2(V_n)$ --- which we call \emph{pairs} --- may consist of either two different or two identical elements of  $V_n$, both cases are needed throughout the paper.
For $ \alpha\beta \in A\subseteq P_2(V_n)$ let $\chi(\alpha\beta)=\chi(\alpha)+\chi(\beta)$ and  $\chi(A)=\sum_{ \alpha\beta\in A}\chi(\alpha\beta)$.
Note that $\chi(\alpha\alpha)$ equals $2$ or $-2$ depending on whether $\sum_{i\in[n]}\alpha(i)$ is even or odd.
We say that a pair $ \alpha\beta \in P_2(V_n)$ is
\begin{enumerate}[]
\item \emph{odd} if $\chi(\alpha)\ne\chi(\beta)$,
\item \emph{even} if $\chi(\alpha)=\chi(\beta)$,
\item \emph{degenerate} if $\alpha =\beta$,
\item \emph{edge}-\emph{pair} if $ \alpha\beta \in E(Q_n)$.
\end{enumerate}
A subset $A\subseteq P_2(V_n)$ is called a \emph{pair}-\emph{set} in $Q_n$ if
\begin{itemize}
\item $\{\alpha ,\beta \}\cap \{\alpha',\beta'\}=\emptyset$ for all distinct $
 \alpha\beta ,\alpha'\beta'\in A$, and
\item if $A\ne\emptyset$ then there exists $ \alpha\beta \in A$ with $\alpha\ne\beta$.
\end{itemize}
Let $\allp_n$ be the set of all pair-sets in $Q_n$.  Given $A\in\allp_n$,
\begin{itemize}
\item $|A|$ denotes the size of $A$,
\item $\|A\|$ denotes the number of odd pairs in $A$,
\item $\bigcup A=\bigcup_{ \alpha\beta \in A}\{\alpha ,\beta \}$.
\end{itemize}
For a positive integer $k$, let $\allp_n^k$ be the set of all pair-sets $A\in\allp_n$ with $|A|\le k$.

We say that a vertex $\alpha\in V_n$ is \emph{encompassed} by a set
$X\subseteq V_n$ if for every edge $ \alpha\beta \in E(Q_n)$ we have $
\beta\in X$. Let $\enco(X)$ be the set of all vertices encompassed by $X$ and $\enc(X)=\enco(X)\setminus X$. Given a pair-set $A$, we set $\enco(A)=\enco(\bigcup A)$, $\enc(A)=\enco(A)\setminus\bigcup A$ and say that a vertex $\alpha$ is  \emph{encompassed by a pair-set} $A$ if $\alpha$ is encompassed by $\bigcup A$.

For a pair-set $A\in\allp_n$ and $i\in[n]$, let $\sigma_i(A)=(n_0,n_1)$ where $n_k$ is the number of $ \alpha,\beta \in A$ with $\alpha (i)=\beta (i)=k$ for $k\in[2]$. Note that $|A|=n_0+n_1+|\{ \alpha,\beta \in A\mid\alpha
(i)\ne\beta (i)\}|$
where $\sigma_i(A)=(n_0,n_1)$. Further, for $k\in[2]$ define
$$\rho_{i=k}(A)=\{\{\rho_{i=k}(\alpha ),\rho_{i=k}(\beta )\}\mid
 \alpha,\beta \in A,\,\alpha (i)=\beta (i)=k\}.$$
Note that 
if $\rho_{i=k}(A)$ is non-empty and consists only of degenerate pairs, then 
$\rho_{i=k}(A)$ is not a pair-set.  
If $\sigma_i(A)=(n_0,n_1)$ and $\rho_{i=k}(A)$ for $k\in[2]$ is a pair-set, then $|\rho_{i=k}(A)|=n_k$.  Conversely, if $A_0,A_1\in\allp_{n-1}$ and $k\in[2]$, we define
$$\begin{aligned}\iota_{i,k}(A_0,A_1)=&\{\{\iota_{i=k}(\alpha ),\iota_{i=k}
(\beta )\}\mid  \alpha,\beta \in A_0\}\,\cup\\
&\{\{\iota_{i=1-k}(\alpha ),\iota_{i=1-k}(\beta )\}\mid \{\alpha
,\beta \}\in A_1\}.\end{aligned}$$
Note that $\iota_{i,k}(A_0,A_1)\in\allp_n$ and
$$\rho_{i=k}(\iota_{i,k}(A_0,A_1))=A_0\text{ while }\rho_{i=1-k}(\iota_{
i=k}(A_0,A_1))=A_1.$$
We say that a~pair-set
$A\in\allp_n$ is
\begin{enumerate}[]
\item \emph{odd} if every pair in $A$ is odd,
\item \emph{balanced} if $\chi (A)=0$,
\item \emph{pure} if $\alpha\ne\beta$ for all $ \alpha,\beta \in A$,
\item \emph{diminishable} if $A$ is odd and
\begin{itemize}
\item
either $|A|\le n-1$ and if $n=4$ then
\begin{itemize}
\item
either $A$ contains an edge pair
\item
or there exists no subset $\bigcup A\subseteq X\subseteq V_4$ such that the  subgraph of $Q_4$ induced by $X$ is isomorphic to $Q_3$,
\end{itemize}
\item
or $|A|=n$, $n\ne 4$, $A$ contains at least two edge-pairs, and $\enc(A)=\emptyset$.
\end{itemize}
\end{enumerate}
Let $\balp_n$, $\oddp_n$, $\dimp_n$ denote the sets of all balanced,  odd, diminishable pair-sets from $\allp_n$, respectively. 
Given a~positive integer $k$, put 
\[
\balp_n^k=\balp_n\cap\allp_n^k, \quad\oddp_n^k=\oddp_n\cap\allp_n^k,\quad\dimp_n^k=\dimp_n\cap\allp_n^k.
\]
Given $A\in\dimp_n$ such that $|A|=n$, we say that $i\in[n]$ is \emph{separating} for $A$ if there exist edge-pairs $ \alpha\beta ,\alpha'\beta'\in A$ with $\alpha (i)=\beta (i)\ne\alpha'(i)=\beta'(i)$.

\subsection{Connectability}
A~pair-set $A\in\allp_n$ is called
 \emph{connectable} if there exists a family $\{P_{\alpha\beta}\mid
 \alpha\beta \in A\}$ of
vertex-disjoint paths in $Q_n$ such that $P_{\alpha\beta}$ is a path between
$\alpha$ and $\beta$ and for every vertex $\gamma\in V_n$ there is exactly one $
 \alpha\beta \in A$ such
that the path $P_{\alpha\beta}$  passes through $\gamma$. Then the family $\{P_{\alpha\beta}\mid  \alpha,\beta \in A\}$
is called a \emph{connector} of $A$.
Observe that if $\{P_{\alpha\beta}\mid  \alpha,\beta \in A\}$ is a connector of a pair-set $A$ and $\alpha\alpha\in A$, then $P_{\alpha\alpha}$ is a singleton path consisting of $\alpha$ because  it is the unique path from $\alpha$ to $\alpha$. Let $\conp_n$ denote the set of all connectable pair-sets in $\allp_n$.

If there exists a vertex $\alpha\in\enc(A)$ for an odd pair-set $A$, then $A$ is not connectable, because 
any path of length $>1$ passing through $\alpha$ visits two distinct pairs of $A$. Note that for every $n\ge 3$ there exists such an odd pair set $A$ with $\enc(A)\ne\emptyset$ and therefore $\Omega_n^n\setminus \Gamma_n\ne\emptyset$ for $n\ge3$. This argument fails for diminishable pairs, because if $\alpha\in\enco(A)$ then $\alpha\in\bigcup A$. 

If $P$ is a path in $Q_n$ between $\alpha$ and $\beta$ and $\gamma ,\delta\in V_n$ are vertices that belong to $P$, then we say that $\gamma$ \emph{is} \emph{closer} \emph{to} $\alpha$ \emph{than} $\delta$ \emph{in} $P$ if the subpath of $P$ between $\alpha$ and $\gamma$ does not contain $\delta$.

For $A,B\in\allp_n$ we write $A\implies B$ if there exist
$ \alpha\beta,\alpha'\beta'\in A$ and $i\in[n]$ such that $\beta\oplus e_i=\beta'$  and $B=(A\setminus \{ \alpha\beta ,\alpha'\beta'\})\cup \{\alpha \alpha'\}$. The transitive and reflexive closure of $\implies$ is denoted by $\overset {*}{\implies}$.
For $i\in[n]$ we say that $A$ is 
\begin{enumerate}[]
\item $i$-\emph{complete} if $\alpha(i)=\beta(i)$ for every $ \alpha,\beta \in A$,
\item an $i$-\emph{completion} of $B$ if $A$ is $i$-complete and $A\overset{*}{\implies}B$.
\end{enumerate}

\section{Preliminaries}
The following lemma summarizes basic properties of notions introduced in the previous section. In particular, its last statement forms the cornerstone of our proof technique.
\begin{lemma}
\label{lemma-simple}
Let $n\ge 1$ be a natural number. Then
\begin{enumerate}[\upshape (1)]
\item \label{lemma-simple:1}
if $A\in\conp_n$ and $A\overset {*}{\implies}B$ then $B\in\conp_
n$;
\item \label{lemma-simple:2}
if $n>1$ and $A,B\in\conp_{n-1}$, then $\iota_{i,k}(A,B)\in\conp_n$ for every $i\in[n]$ and $k\in[2]$,
\item \label{lemma-simple:3}
if $A$ is balanced and $B\overset{*}{\implies}A$ then $B$ is balanced,
\item \label{lemma-simple:4}
if $A\in\conp_{n}$ then $A$ is balanced,
\item \label{lemma-simple:5}
if $A\in\oddp_n$ then $A$ is balanced,
\item \label{lemma-simple:6}
\label{lemma-simple:separate}
if $A\in\allp_n$ consists of two or three edge-pairs, then there exists $i\in[n]$ such that $i$ is separating for $A$ and $A$ is $i$-complete unless $A$ consists of three edge-pairs lying on a 6-cycle,
\item
\label{lemma-simple:completion}
If $B\in\allp_n$ is an $i$-completion of $A$ and $\rho_{i=0}(B),\rho_{i=1}(B)\in\conp_{n-1}$, then $A\in\conp_n$. 
\end{enumerate}
\end{lemma}
\begin{proof}
Properties \eqref{lemma-simple:1}-\eqref{lemma-simple:3}, \eqref{lemma-simple:5} and \eqref{lemma-simple:completion} are straightforward corollaries of  the corresponding definitions.  

To verify \eqref{lemma-simple:4}, let $A\in\conp_n$ and $\mathcal P=\{P_{\alpha\beta}\mid \alpha\beta\in A\}$ be a connector of $A$. Observe that as $\p(\gamma\gamma')=0$  for every edge $\gamma\gamma'\in E(Q_n)$, it follows that for every path $P_{\alpha\beta}\in\mathcal P$ we have
$\p(V(P_{\alpha\beta}))=\p(\alpha\beta)$. Consequently,
\[
\p(\mathcal A)=\sum_{\alpha\beta\in A}\p(\alpha\beta)=\sum_{\alpha\beta\in A}\p(V(P_{\alpha\beta}))=\sum_{\gamma\in V(Q_n)}\p(\gamma)=0
\]
and therefore $A$ is balanced.

To prove \eqref{lemma-simple:separate}, we argue by induction on $n$. Since the case that $n\le2$ is obvious, assume that $n\ge3$ and put $D=\{\Delta(\alpha,\beta)\mid\alpha\beta\in A\}$. If $D=[n]$ then it must be the case that $n=3$ and $A$ consists of three edge-pairs lying on a 6-cycle. Otherwise select an arbitrary $i\in[n]\setminus D$ and note that then $A$ is $i$-complete. If $i$ is separating for $A$, we are done. Otherwise there is $k\in[2]$ such that $\alpha(i)=\beta(i)=k$ for all $\alpha\beta\in A$ and hence by the induction hypothesis, there exists $i'\in[n]\setminus\{i\}$ satisfying the conclusion of \eqref{lemma-simple:separate}.
\end{proof}

Next we recall several previous results on hypercubes that shall be useful later.

\begin{proposition}
\label{old-results-prop}
Let $n\ge 1$ be a natural number. Then
\begin{enumerate}[\upshape(1)]
\item{\rm\cite[Proposition 2.3]{H}}
\label{old-results-prop-part1}
a singleton pair-set $A$ belongs to $\conp_n$ if and only if $A$ is odd; \item{\rm\cite[Corollary 4]{LW}}
\label{old-results-prop-part2}
if $A=\{ \alpha\beta ,\gamma\gamma\}\in\balp_n$, then $A\in\conp_n$;
\item{\rm\cite[Lemma 3.3]{D}}
\label{old-results-prop-part3}
$\oddp_n^2\subseteq\conp_n$;
\item{\rm\cite[Lemma 3.3]{D}}
\label{old-results-prop-part3a}
if $\{ \alpha\beta ,\gamma \delta \}\in\oddp_n$ such that $\gamma \delta$ is an edge-pair, then $\{ \alpha\beta ,\gamma\gamma,\delta\delta\}\in\conp_n$ unless $n=3$,  $\alpha,\beta$ is an edge-pair and $d( \alpha\beta ,\gamma \delta)=2$;
\item{\rm\cite[Corollary 10]{CK}}
\label{old-results-prop-part5}
if $A\in\balp^2_n$ is pure with $||A||=0$ and $n\ge 4$, then $A\in\conp_n$;
\item{\rm\cite[Lemmas 12-14]{CK}}
\label{old-results-prop-part6}
if $A\in\oddp^3_n$ and $n\ge 5$, then $A\in\conp_n$;
\item{\rm\cite[Lemmas 3.6 and 3.11]{CG}, \cite[Theorem 6.1]{CGL}}
\label{old-results-prop-part7}
if $A\in\balp_n^3$
with  $n\ge 4$
and $\|A\|=1$, then $A\in\conp_n$;
\item{\rm\cite[Lemmas 3.13, 4.3, 5.6, and 5.7]{CG}, \cite[Theorems 3.1 and 4.1]{CGGL10}}
\label{old-results-prop-part8}
if $A\in\balp_n^4$ is not pure and $n\ge 5$,
then $A\in\conp_n$;
\item{\rm\cite[Lemma 5.6]{CG}}
\label{old-results-prop-part9}
 if $A\in\balp_n^4$, $\|A\|=2$, $n\ge4$, and $A$ has two degenerate pairs then $A\in\conp_n$;
\item{\rm\cite[Lemma 4.6]{CG}}
\label{old-results-prop-part10}
 if $A\in\balp_n^5$, $n\ge5$ and $A$ has four degenerate pairs then $A\in\conp_n$;
 \item{\rm\cite[Lemma~3.4]{DGK17}}
 \label{old-results-prop-part13}
 if $A$ is odd, $|A|=n-1$ and $n\ge4$, then there exists an $i$-completion of $A$ for some $i\in[4]$.
\end{enumerate}
\end{proposition}


\section{Small dimensions}
We say that pair-sets $A,B\in\allp_n$ are \emph{isomorphic} if there
exists an automorphism $f$ of a hypercube $Q_n$ such that
$B=\{\{f(\alpha),f(\beta)\}\mid \alpha,\beta \in A\}$. 

\begin{figure}[!h]
\centering
\hspace{-1cm}%
\scalebox{1}{\includegraphics{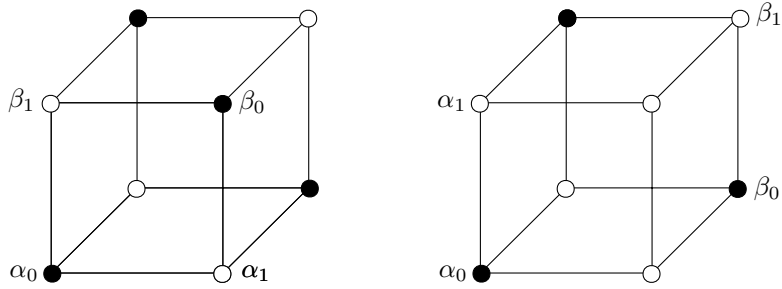}}
\begin{minipage}[t]{.23\linewidth}
\centering
\end{minipage}%
\begin{minipage}[t]{.38\linewidth}
\centering
\end{minipage}
\caption{The only non-connectable balanced pair-sets $C_0=\{\alpha_0\beta_0,\alpha_1\beta_1\}$ and $C_1=\{\alpha_0\beta_0,\alpha_1\beta_1\}$ in $Q_3$\label{fig:Fig1}}
\end{figure}

\begin{proposition}
\label{dimen-3}
Let $A\in\balp_3$. Then
\begin{enumerate}[\upshape(1)]
\item
\label{dimen-3-1}
if $|A|=2$ then $A$ is connectable if and only if $A$ is not isomorphic to the pair-sets $C_0$ and $C_1$
on \Cref{fig:Fig1},
\item
\label{dimen-3-2}
if $A$ is diminishable then $A$ is connectable.
\end{enumerate}
\end{proposition}
\begin{proof}
Case that $|A|\le2$ follows from parts \eqref{old-results-prop-part1} -- \eqref{old-results-prop-part3} of \Cref{old-results-prop} except the case that $A$ consists of two non-degenerated pairs which may be obtained by a direct inspection. If $A$ is a~diminishable pair-set of size three,, then it contains at least two edge-pairs. If $A$ consists of three edge-pairs, then the condition $\enc(A)\ne\emptyset$ from the definition of diminishability impúlies that these pairs do not lie on a 6-cycle. Hence by \Cref{lemma-simple}\,\eqref{lemma-simple:separate}, $Q_3$ may be partitioned into two edge-pairs and a 4-cycle including the third edge pair, which means that $A\in\conp_3$ in this caase. It remains to verify the case that $A$ consists of two edge-pairs and a pair of diametrical vertices, which may be obtained by a direct inspection.
\end{proof}
\begin{proposition}
\label{dimen-4}
Let $A\in\oddp_4$ be a pair-set. Then
\begin{enumerate}[\upshape(1)]
\item
\label{dimen-4-1}
if $|A|\le 3$ then $A$ is connectable if and only if $A$ is not isomorphic to the pair-set $C_2$ on~\Cref{fig:Fig2},
\item
\label{dimen-4-2}
if $|A|=4$, $A$ contains at least three edge pairs and $\enc(A)=\emptyset$, then $A$ is connectable.
\end{enumerate}
\end{proposition}

\begin{figure}
\centering
\hspace{-1cm}%
\begin{minipage}[t]{0.5\linewidth}
\centering
\scalebox{1}{\includegraphics{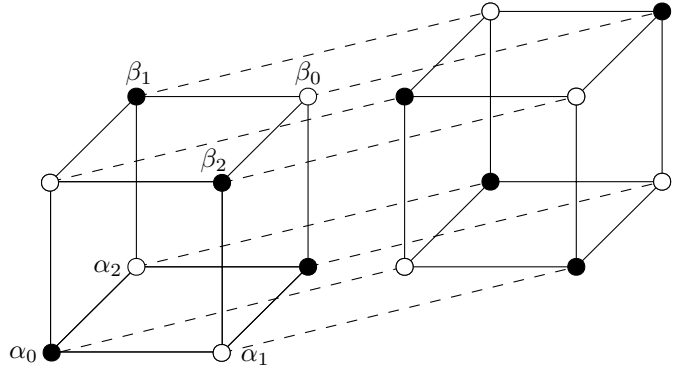}}
\end{minipage}
\caption{The only non-connectable odd pair-set $C_2=\{\alpha_i\beta_i\mid i\in[[3]\}$ in $Q_4$\label{fig:Fig2}}
\end{figure}

\begin{proof}
First note that to verify part \eqref{dimen-4-1}, it suffices to assume that $|A|=3$, as the case that $|A|<3$ follows from parts \eqref{old-results-prop-part1} and \eqref{old-results-prop-part3} of \Cref{old-results-prop}.
By a direct verification we obtain that $C_2$ is not connectable. Assume that $A=\{\alpha_i\beta_i\mid i=0,1,2\}\in\oddp_4$ is not isomorphic to $C_2$. If there
exists $i\in [4]$ such that $\sigma_i(A)=(n_0,n_1)$ and $n_0,n_1>0$ then, by  \Cref{old-results-prop}\,\eqref{old-results-prop-part13}, we find an $i$-completion $B$ of $A$ such that $B\implies A$, $|\rho_{i=0}(B)|,|\rho_{i=1}(B)|\le 2$ and, by Proposition~\ref{old-results-prop}\,\eqref{old-results-prop-part3}, both $|\rho_{i=0}(B)|$ and $|\rho_{i=1}(B)|$ are connectable and thus $A$ is connectable.
Thus we can assume that $n_0=0$ or $n_1=0$. Since $A$ is odd and $n=4$,
there exists $i\in[4]$ such that if $\sigma_i(A)=(n_0,n_1)$ then either $n_0>0$ or $n_1>0$. We can without loss of generaility assume that $n_0>0$. If $n_0=3$ then there must be an edge-pair, for otherwise $A$ is isomorphic to $C_2$, contrary to our assumption.
It follows there exists $j\in[3]$ and $k\in[1]$ such that $0<m_k<3$ for $\sigma_j(\rho_{i=0}(A))=(m_0,m_1)$. Hence we can assume that there exists $i\in[4]$ such that $\sigma_i(A)=(1,0)$ or
$\sigma_i(A)=(2,0)$.

First consider the case that $\sigma_i(A)=(2,0)$ for some $i\in [4]$ and
$\alpha_0(i)=\alpha_2(i)=\alpha_3(i)=\beta_0(i)=\beta_1(i)=0\ne\beta_
2(i)$.  First assume that
$\alpha_0\beta_0\in E(Q_4)$.  By an easy calculation, there exists $
\gamma_2\in V_4$
such that $\alpha_2\gamma_2\in E(Q_4)$, $\gamma_2(i)=0$ and no vertex of
$V_3\setminus\rho_{i=0}(\bigcup A\cup \{\gamma_2\})$ is encompassed by $
\rho_{i=0}(\bigcup A\cup \{\gamma_2\})$.  By
Proposition~\ref{dimen-3}\,\eqref{dimen-3-2},
$$A^{\prime\prime}=\{\rho_{i=0}(\alpha_0)\rho_{i=0}(\beta_0),\rho_{i=0}(\alpha_1)\rho_{i=0}(\beta_1),\rho_{i=0}(\alpha_
2)\rho_{i=0}(\gamma_2)\}$$
is connectable. Then $\chi (\rho_{i=0}(\gamma_2))\ne\chi (\rho_{i
=1}(\beta_2))$ and thus
$$A^{\prime\prime\prime}=\{\rho_{i=0}(\gamma_2)\rho_{i=1}(\beta_2)\}$$
is connectable and hence  $A$ is connectable as well. Thus we can restrict to the case that
$ \alpha_0\beta_0 $, $ \alpha_1\beta_1 \notin E(Q_4)$. 

Set  $\alpha=\rho_{i=0}({\alpha_2})$,  $A'=\rho_{i=0}(A)$ and observe that  by \Cref{old-results-prop}\,\eqref{old-results-prop-part3}, $A'\in\Gamma_3$ and hence there exists a connector $\{P_{\gamma\gamma'}\mid \gamma\gamma'\in A'\}$  of $A'$. Let $P_{\kappa\kappa'}$ be the path passing through $\alpha$. Then there must be a~subpath $(\zeta ,\alpha ,\nu, \zeta')$ of $P_{\kappa\kappa'}$. Without loss of generality assuming that $\zeta$ is closer to $\kappa$ than $\zeta'$ in $P_{\kappa\kappa'}$, set  
\[
A^{\prime\prime}=(A'\setminus \{ \kappa\kappa' \})\cup \{ \kappa\zeta , \alpha\nu , \zeta'\kappa' \}
\]
and note that $A''\in\conp_{n-1}$.
Further set
\[
A^{\prime\prime\prime}=\{ \zeta\zeta' ,\nu\rho_{i=1}(\beta_2)\}
\]
and observe that if $\{\zeta,\zeta' \}\cap\{\nu,\rho_{i=1}(\beta_2)\}=\emptyset$, then $A'''$ is an odd pair-set and therefore connectable by \Cref{old-results-prop}\,\eqref{old-results-prop-part3}. Otherwise it must be the case that $\zeta'=\rho_{i=1}(\beta_2)$. If  $P_{\kappa\kappa'}$ is a path of length $>3$, it must contain a subpath  $(\zeta ,\alpha,\alpha',\alpha'',\nu, \zeta')$. Then $A'''$ is an odd pair-set and therefore $A'''\in\conp_{n-1}$. It remains to deal with the case that $P_{\kappa\kappa'}=(\kappa ,\alpha ,\nu, \kappa')$ while  $\kappa'=\rho_{i=1}(\beta_2)$. Here set  
\[
A^{\prime\prime}=(A'\setminus \{ \kappa\kappa' \})\cup \{\kappa\kappa, \alpha\nu ,\kappa'\kappa'\}
\]
and note that $A''\in\conp_3$. Further set
\[
A^{\prime\prime\prime}=\{\kappa\nu,\alpha\rho_{i=1}(\beta_2)\}
\]
and observe that  as $\kappa\kappa'$ is not an edge-pair by our assumption, we have $\kappa\rho_{i=1}(\beta_2)\not\in E(Q_3)$ while $\kappa\alpha,\alpha\nu\in E(Q_3)$  and therefore $A'''$ is not isomorphic to the sets $C_0, C_1$ on \Cref{fig:Fig1}. Hence $A'''\in\conp_{3}$ by \Cref{dimen-3}\,\eqref{dimen-3-1}.
In any case, since
$\iota_{i,0}(A^{\prime\prime},A^{\prime\prime\prime})\overset{*}{\implies}A$, by Lemma~\ref{lemma-simple} we can conclude that 
$A$ is connectable. Thus if there exists $i\in [4]$ with $\sigma_i(A)=(2,0)$ then $A$ is connectable.

Secondly assume that  $\sigma_i(A)=(1,0)$ for some $i\in[4]$. Let $\alpha_0(i)=\beta_0(i)=\alpha_1(i)=\alpha_2(i)=0\ne\beta_1(i)=\beta_2(i)$. Since $\rho_{i=0}(\alpha_0)\rho_{i=0}(\beta_0)$ is an odd pair in $Q_3$, there exists a cycle $C$ of length $6$ in $Q_3$  such that $\rho_{i=0}(\alpha_0),\rho_{i=0}(\beta_0)\notin V(C)$.   Hence $\{\rho_{i=0}(\alpha_1),\rho_{i=0}(\alpha_2),\gamma_1,\gamma_2,\kappa',\kappa\}= V(C)$ for some $\gamma_1,\gamma_2,\kappa',\kappa\in V_3$  such that
$\alpha_1\gamma_1$, $\alpha_2\gamma_2$, $\kappa\kappa'$ are edges of $C$. Then
$$A^{\prime\prime}=\{\rho_{i=0}(\alpha_0)\rho_{i=0}(\beta_0),\rho_{i=0}(\alpha_1)\gamma_1,\rho_{i=1}(\alpha_2)\gamma_2\}$$ is connectable by Proposition~\ref{dimen-3}\,\eqref{dimen-3-2}. If
$\{\gamma_1,\gamma_2\}\cap \{\rho_{i=1}(\beta_1),\rho_{i=1}(\beta_2)\}
=\emptyset$ then, by Proposition~\ref{old-results-prop}\,\eqref{old-results-prop-part3},
$A^{\prime\prime\prime}=\{\gamma_j\rho_{i=1}(\beta_j)
\mid j=1,2\}$ is a connectable pair-set because $A'''$ is odd and hence $A$ is connectable.  Thus we can restrict to the case that 
$\{\gamma_1,\gamma_2\}\cap \{\rho_{i=1}(\beta_1),\rho_{i=1}(\beta_2)\}\ne\emptyset$. Note that this assumption implies that 
\begin{enumerate}[\upshape(a)]
\item $\chi (\alpha_1)\ne\chi (\alpha_2)$,
\item $\{\rho_{i=0}(\alpha_1),\rho_{i=0}(\alpha_2)\}\ne\{\rho_{i=1}(\beta_1),\rho_{i=1}(\beta_2)\}$,
\item \label{dimen-4-1:c}
$\dist(\alpha_1,\beta_2)=2$ or $\dist(\alpha_2,\beta_1)=2$. 
\end{enumerate}
Set $A'=\rho_{i=0}
(\{ \alpha_0\beta_0 ,\alpha_1\alpha_2\})$, $\hat\alpha_j=\rho_{i=0}({\alpha })$, $\hat\beta_j=\rho_{i=0}({\beta})$ for both $j\in\{1,2\}$ and note that  $A'\in\Gamma_{3}$ by Proposition~\ref{old-results-prop}\,\eqref{old-results-prop-part3} and therefore there is a connector $\{P_{\gamma\gamma'}\mid \gamma\gamma'\in A'\}$ of $A'$.  Then $P_{\hat\alpha_1\hat\alpha_2}$ contains an edge $\zeta\nu$ such that 
$\zeta$ is closer to $\hat\alpha_1$ than $\nu$ on this path and 
\begin{enumerate}[\upshape(A)]
\item
\label{dimen-4-1:A}
$\p(\zeta)\ne\p(\hat\alpha_1)$ and $\{\zeta,\nu\}\cap\{\hat\beta_1,\hat\beta_2\}=\emptyset$; or
\item
\label{dimen-4-1:B}
$\p(\zeta)=\p(\hat\alpha_1)$ and $|\{\zeta,\nu\}\cap\{\hat\beta_1,\hat\beta_2\}|=1$; or
\item
\label{dimen-4-1:C}
$P_{\hat\alpha_1\hat\alpha_2}=(\hat\alpha_1,\hat\alpha_2)=(\zeta,\nu)$ and $\{\zeta,\nu\}\cap\{\hat\beta_1,\hat\beta_2\}=\emptyset$.
\end{enumerate}
Set
\begin{align*}
A^{\prime\prime}&=
(A'\setminus \{\hat\alpha_1\hat\alpha_2\})\cup \{\hat\alpha_1\zeta,\nu\hat\alpha_2\},\\
A^{\prime\prime\prime}&=\{\zeta\hat\beta_1,\nu\hat\beta_2\}.
\end{align*}
Since $\{P_{\gamma\gamma'}\mid \gamma\gamma'\in A'\}$ is a connector of $A'$, we can conclude that in all three cases $A^{\prime\prime}\in\Gamma_{n-1}$. If \eqref{dimen-4-1:A} or \eqref{dimen-4-1:B} occurs,  $A'''$ is an odd pair-set or non-pure balanced pair-set and therefore connectable by parts \eqref{old-results-prop-part3} or \eqref{old-results-prop-part2} of \Cref{old-results-prop}, respectively. If  \eqref{dimen-4-1:C} occurs, $A'''$ is a pure balanced pair-set and therefore connectable by \Cref{dimen-3}\,\eqref{dimen-3-1} unless $A'''$ is isomorphic to one of pair-sets $C_1,C_2$ on \Cref{fig:Fig1}. In all these cases we have
$\iota_{i,0}(A^{\prime\prime},A^{\prime\prime\prime})\overset{*}{\implies}A$ and therefore by Lemma~\ref{lemma-simple} we can conclude that $A\in\Gamma_4$. 

It remains to deal with the case when \eqref{dimen-4-1:C} holds and $A'''$ is isomorphic to one of forbidden configurations on \Cref{fig:Fig1}. Recall that our assumption \eqref{dimen-4-1:c}  implies that then $A'''$ must be isomorphic to the pair-set $C_0$. If follows that $\{\alpha_1,\alpha_2\}$ and $\{\beta_1,\beta_2\}$ are edges on a 6-cycle. Put $D=\{\Delta(\gamma,\gamma')\mid \gamma,\gamma'\in \{\alpha_1,\alpha_2,\beta_1,\beta_2\}\}$, $D'=\Delta(\alpha_0,\beta_0)$ and note that $|D|=3$ and $i\in D\setminus D'$. Hence if $|D'|=3$, then there is $j\in D' \setminus D$ such that $\sigma_j(A)\in\{(2,0), (0,2)\}$ which leads to the case settled above. Otherwise we have $|D'|=1$, i.e. $ \alpha_0\beta_0 $ is an edge-pair. By \Cref{lemma-simple}\,\eqref{lemma-simple:separate} then there is a $j\in[4]$ separating for $\{ \alpha_0\beta_0 , \alpha_1\beta_1 \}$, which means that $ \alpha_0\beta_0 $ and $ \alpha_1\beta_1 $ lie on two disjoint 4-cycles that partition $Q_3$. 
Consequently, in this case there is a~connector of $A'$ consisting of two paths of length three, which leads to cases \eqref{dimen-4-1:A},\eqref{dimen-4-1:B} settled above. The proof of Proposition~\ref{dimen-4}\,\eqref{dimen-4-1} is complete.

To prove Proposition~\ref{dimen-4}\,\eqref{dimen-4-2}, let $A=\{\alpha_j\beta_j\mid j\in[4]\}$ be an odd pair-set in $\oddp_4$ such that $\alpha_j\beta_j$ is an edge pair for $j\in[3]$. First assume that 
\begin{enumerate}[($\ast$)]
\item\label{dimen-4-1:star}
 there is $i\in[4]$ such that
$\alpha_0(i)=\beta_0(i)=\alpha_2(i)=\beta_2(i)=0\ne\alpha_1(i)=\beta_1(i)$.
\end{enumerate}
Let $\sigma_i(A)=(n_0,n_1)$, then $n_0+n_1\in\{3,4\}$.

If $n_0+n_1=4$ then $A$ is $i$-complete. If both $\rho_{i=0}(A)$ and $\rho_{i=1}(A)$ are diminishable, then they are connectable by Proposition~\ref{dimen-3}  and therefore $A\in\conp_4$ by \Cref{lemma-simple}\,\eqref{lemma-simple:completion}. Otherwise it must be the case that $\enc(A')=\{\gamma,\gamma'\}$ where $A'=\rho_{i=0}(A)$. Assuming that $\gamma$ is encompassed by $\rho_{i=0}(\alpha_0),\rho_{i=0}(\alpha_2)$, and  $\rho_{i=0}(\alpha_3)$  (which means that $\gamma'$ is encompassed by the remaining vertices of $\bigcup A'$)
, set
\[
A''=(A'\setminus\{\rho_{i=0}(\alpha_3)\rho_{i=0}(\beta_3)\})\cup\{\rho_{i=0}(\alpha_3)\gamma,\rho_{i=0}(\beta_3)\gamma' \}
\]
and note that then $A''$ consists of four edge-pairs and therefore $A''\in\conp_3$.  Recall that $\enc(A)=\emptyset$ by our assumption and therefore $\{\gamma,\gamma'\}\cap\rho_{i=1}(A)=\emptyset$. It follows that 
\[
A'''=\{\rho_{i=1}(\alpha_1)\rho_{i=1}(\beta_1),\gamma\gamma'\}
\]
is an odd pair set and therefore connectable by \Cref{dimen-3}. Since we have
$\iota_{i,0}(A^{\prime\prime},A^{\prime\prime\prime})\overset{*}{\implies}A$, by Lemma~\ref{lemma-simple} we can conclude that $A\in\Gamma_4$. 

If $n_0+n_1=3$ then we can assume that $\alpha_3(i)=0\ne\beta_3(i)$. Observe that there must be a $\gamma\in V_4\setminus\bigcup A $ such that $\gamma(i)=0$, $\p(\gamma)\ne\p(\alpha_3)$,$\gamma\oplus e_i\not\in A$ and $\enc(\rho_{i=0}(A)\cup\{\gamma\})=\emptyset$. Set 
\[
B=(A\setminus\{\alpha_3\beta_3\})\cup\{\alpha_3\gamma,\gamma\oplus e_i\beta_3\}.
\]
Then $|B|=5$, $\sigma_i(B)=(3,2)$, $B\implies A$, both $\rho_{i=0}(B)$ and $\rho_{i=1}(B)$ are connectable by Proposition~\ref{dimen-3} and hence $A\in\conp_4$.

If  \hyperref[dimen-4-1:star]{($\ast$)} 
does not apply, by \Cref{lemma-simple}\,\eqref{lemma-simple:separate} it must be the case that edge-pairs $\{\alpha_j\beta_j\mid j\in[3]\}$ lie on a~6-cycle while $\dist(\alpha_3,\beta_3)=3$. Then $[4]\setminus\{\Delta(\alpha_j,\beta_j)\mid j\in[3])\}=\{i\}$ and we can assume that $\alpha_j(i)=\beta_j(i)=0$ for all $j\in[3]$. Put $A'=\rho_{i=0}(A)$ and consider three subcases. If $\sigma_i(A)=(4,0)$, set
\begin{align*}
A''&=(A'\setminus\{\rho_{i=0}(\alpha_3)\rho_{i=0}(\beta_3)\})
\cup\{\rho_{i=0}(\alpha_3)\rho_{i=0}(\alpha_3),\rho_{i=0}(\beta_3)\rho_{i=0}(\beta_3)\},\\
A'''&=\{\rho_{i=0}(\alpha_3)\rho_{i=0}(\beta_3)\}
\end{align*}
and note that $A''$ consists of three edge-pairs and two degenerated pairs covering $V_3$ and therefore $A''\in\conp_3$ while $A'''\in\conp_3$ by \Cref{dimen-3}. 

If $\sigma_i(A)=(3,1)$, then $V_3\setminus\bigcup A'=\{\gamma,\gamma'\}$ and we can assume that $\{\rho_{i=0}(\alpha_j)\mid j\in[3]\}$ encompasses $\gamma$ while $\{\rho_{i=0}(\beta_j)\mid j\in[3]\}$ encompasses $\gamma'$. It follows that 
\[
A''=(A'\setminus\{\rho_{i=0}(\alpha_2)\rho_{i=0}(\beta_2)\})\cup\{\rho_{i=0}(\alpha_2)\gamma,\rho_{i=0}(\beta_2)\gamma'\}
\]
is connectable. Since we assumed that $\enc(A)=\emptyset$, it follows that $\{\gamma,\gamma'\}\cap\{\rho_{i=1}(\alpha_3),\rho_{i=1}(\beta_3)\}=\emptyset$ and therefore
\[
A'''=\{\gamma\gamma',\rho_{i=1}(\alpha_3)\rho_{i=1}(\beta_3)\}
\]
is an odd pair-set and therefore $A'''\in\conp_3$ by \Cref{dimen-3}.

If $\sigma_i(A)=(3,0)$, we can assume that $\alpha_3(i)=0\ne\beta_3(i)$ and that $\rho_{i=0}(\alpha_3)$ is encompassed by $\{\rho_{i=0}(\alpha_j)\mid j\in[3]\}$ while $\rho_{i=1}(\beta_3)=\rho_{i=0}(\beta_2)$. Then $V_3\setminus\rho_{i=0}(\bigcup A)$ consists of a sole $\gamma$ which is encompassed by $\{\rho_{i=0}(\beta_j)\mid j\in[3]\}$ and therefore
\[
A''=(A'\setminus\{\rho_{i=0}(\alpha_2)\rho_{i=0}(\beta_2)\})\cup\{\rho_{i=0}(\alpha_2)\rho_{i=0}(\alpha_2),\rho_{i=0}(\beta_2)\gamma,\rho_{i=0}(\alpha_3)\rho_{i=0}(\alpha_3)\}
\]
is connectable. Furthermore,
\[
A'''=\{\rho_{i=0}(\alpha_3)\rho_{i=1}(\beta_3),\gamma\rho_{i=0}(\alpha_2)\}
\]
is a balanced pair-set where $\dist(\rho_{i=0}(\alpha_3),\gamma)=3$ while $\dist(\rho_{i=1}(\beta_3),\rho_{i=0}(\alpha_2))=1$. It follows  that $A'''$ is not isomorphic to any configuration on \Cref{fig:Fig1} and therefore $A'''\in\conp_3$ by \Cref{dimen-3}.

Since in all cases we have
$\iota_{i,0}(A^{\prime\prime},A^{\prime\prime\prime})\overset{*}{\implies}A$, by Lemma~\ref{lemma-simple} we can conclude that $A\in\Gamma_4$ and the proof is complete. 
\end{proof}

It should be noted that the assumptions of Proposition~\ref{dimen-4}\,\eqref{dimen-4-2} are the best possible in the following sense: by a computer search we identified 53 non-connectable pair-sets $A\in\oddp_4$ such that $|A|=4$, $A$ contains two edge pairs and $\enc(A)=\emptyset$.

\end{document}